\pdfoutput=1

\documentclass[11pt]{article}
\usepackage[utf8]{inputenc}
\usepackage[letterpaper, margin=1in]{geometry}
\usepackage{palatino}

\usepackage{microtype}          %
\usepackage[T1]{fontenc}
\usepackage[full]{textcomp}     %

\usepackage[usenames,dvipsnames]{xcolor}

\usepackage[style=alphabetic,natbib=true,maxalphanames=10,maxnames=10,sortcites=true,doi=false,url=false]{biblatex} %
\bibliography{dsg}

\usepackage{graphicx}
\usepackage{amsmath,amsthm,amsfonts,amssymb,amscd}
\usepackage{bbm}
\usepackage{mathtools}
\usepackage{enumitem}
\usepackage{xfrac}
\usepackage{mathrsfs}
\usepackage{stmaryrd}
\usepackage{theoremref}
\usepackage[normalem]{ulem}
\usepackage{mleftright}

\usepackage[linesnumbered,ruled,vlined]{algorithm2e}
\usepackage{framed}

\usepackage{tago}

\usepackage{tikz}
\usetikzlibrary{shapes,backgrounds}

\usepackage[symbol]{footmisc}

\definecolor{purp}{HTML}{9483C0}

\SetCommentSty{mycommfont}
\SetArgSty{textnormal}

\usepackage{etoolbox}
\AtBeginEnvironment{align}{\setcounter{equation}{0}}

\setlist[enumerate]{nosep,topsep=0pt,label=\alph*.}
\setlist[itemize]{nosep,topsep=0pt}

\let\originalleft\left
\let\originalright\right
\renewcommand{\left}{\mathopen{}\mathclose\bgroup\originalleft}
\renewcommand{\right}{\aftergroup\egroup\originalright}

\newtheorem{thm}{Theorem}[section]
\newtheorem{corollary}[thm]{Corollary}
\newtheorem{lemma}[thm]{Lemma}

\newcommand{\comment}[1]{}

\newcommand{\arc}[2]{\overrightarrow{#1#2}}
\newcommand{\alldeg}[1]{\delta\left(#1\right)}
\newcommand{\bigO}[1]{\mathcal{O}\left(#1\right)}
\newcommand{\bigOmega}[1]{\Omega\left(#1\right)}
\newcommand{\bigTheta}[1]{\Theta\left(#1\right)}
\newcommand{\code}[1]{{\ttfamily #1}}
\newcommand{\densityopt}{\rho^\star}
\newcommand{\eps}{\varepsilon}

\newcommand{\incut}[1]{\delta^-\left(#1\right)}
\newcommand{\indeg}[1]{d^-\left(#1\right)}
\newcommand{\ita}[1]{{\itshape #1}}
\newcommand{\lb}[1]{\varphi(#1)}
\newcommand{\lab}[1]{\varphi(#1)}
\newcommand{\labc}[2]{\varphi(#1|#2)}
\newcommand{\ld}[1]{\ell\left(#1\right)}
\newcommand{\ldA}[1]{\ell'\left(#1\right)}
\newcommand{\ldB}[1]{\ell''\left(#1\right)}
\newcommand{\ldT}[1]{\ell_T\left(#1\right)}
\newcommand{\load}[1]{\ld{#1}}
\newcommand{\loadT}[1]{\ldT{#1}}
\newcommand{\lpopt}{\mathsf{OPT}_{\mathsf{LP}}}
\newcommand{\lv}[1]{\mathcal{L}\left(#1\right)}
\newcommand{\OPT}{\mathsf{OPT}}
\newcommand{\outcut}[1]{\delta^+\left(#1\right)}

\newcommand{\poly}{\mathrm{poly}}
\newcommand{\Rgeo}{\mathbb{R}_{\geq 1}}
\newcommand{\Rgez}{\mathbb{R}_{\geq 0}}
\newcommand{\Rgz}{\mathbb{R}_{>0}}
\newcommand{\rulesep}{\unskip\ \vrule\ }
\DeclarePairedDelimiterX\set[2]{\{}{\}}{\,#1 \;\delimsize|\; #2\,}
\newcommand{\softO}[1]{{\tilde{\mathcal{O}}}\left(#1\right)}

\newcommand{\sta}{\text{ s.t. }}
\newcommand{\sz}[1]{\left\vert#1\right\vert}
\newcommand{\union}{\cup}
\newcommand{\wt}[1]{w\left(#1\right)}
\newcommand{\zo}{\{0,1\}}

\newcommand{\alphamore}{(1+\alpha)}
\newcommand{\apxless}{(1-\bigO{\eps})}

\newcommand{\epsless}{(1-\eps)}
\newcommand{\epsmore}{(1+\eps)}

\title{Approximate Fully Dynamic Directed Densest Subgraph}
\author{Richard Li \\ \small Purdue University \\ \texttt{richzli@purdue.edu} \and Kent Quanrud\footnote{Supported in part by NSF grant CCF-2129816.} \\ \small Purdue University \\ \texttt{krq@purdue.edu}}
\date{}

\begin{document}

\maketitle

\begin{abstract}
    We give a fully dynamic algorithm maintaining a $\epsless$-approximate directed densest subgraph in $\softO{\log^3(n)/\eps^6}$ amortized time or $\softO{\log^4(n)/\eps^7}$ worst-case time per edge update, based on earlier work by Chekuri and Quanrud \cite{CQ22}. This result improves on earlier work done by Sawlani and Wang \cite{SW20}, which guarantees $\bigO{\log^5(n)/\eps^7}$ worst case time for edge insertions and deletions.
\end{abstract}

\section{Introduction}

The densest subgraph problem (DSG) and extensions of DSG are well-studied problems in both theory and practice. In the simplest version, the input is an undirected graph $G$. Formally, the \ita{density} of $S$ in $G$ is defined as $\rho_G(S)=\frac{\sz{E(S)}}{\sz{S}}$, and one can interpret this value as (half of) the average degree of the induced subgraph of $S$. The densest subgraph problem asks us to find a subgraph of maximum density.

Dense subgraphs have practical applications in a myriad of fields, including (but not limited to) social network analysis and community detection \cite{DBLP:journals/tkde/ChenS12, DBLP:journals/cn/KumarRRT99}, web crawling and data mining \cite{DBLP:conf/vldb/GibsonKT05}, and even computational biology \cite{DBLP:conf/ismb/HuYHHZ05}. For example, a subgraph with high density in a social network may indicate communities of people with similar interests. A subgraph with high density in a hyperlink graph of the Internet may indicate webpages containing similar topics. A subgraph with high density in a protein interaction graph may indicate a protein complex.

Goldberg \cite{Gol84} first introduced the DSG problem in 1984, and gave a polynomial-time solution using $\bigO{\log (n)}$ computations of max flow. This was later improved by Gallo et al. \cite{DBLP:journals/siamcomp/GalloGT89} to use only $\bigO{1}$ instances of parametric max flow. (We note that while max flow algorithms have seen significant breakthroughs over the last few years, notably Chen et al.'s \cite{DBLP:conf/focs/ChenKLPGS22} 2022 nearly-linear time algorithm, these are not currently considered to be practical algorithms.) In 2000, Charikar \cite{Cha00} gave a linear-time $1/2$-approximation for the density based on an exact LP formulation.

With some applications of DSG having billions of nodes and dynamically changing connections, recent focus has been on DSG in approximation, online/streaming, parallel, and dynamic regimes. In 2014, Bahmani et al. \cite{DBLP:conf/waw/BahmaniGM14} discovered a $\bigO{\log n/\eps}$-pass streaming algorithm for a $(1/2-\eps)$-approximate densest subgraph. Bhattacharya et al. \cite{DBLP:conf/stoc/BhattacharyaHNT15} improved this result in 2015 with a $(1/2-\eps)$-approximate one-pass algorithm, and separately developed a $(1/4-\eps)$-approximate fully dynamic algorithm (supporting edge insertions and deletions) with polylogarithmic update time. Recently, Sawlani and Wang \cite{SW20} discovered an arbitrarily exact $\epsless$-approximation in the fully dynamic setting in $\bigO{\log^4(n)/\eps^6}$ worst-case time per update, and Chekuri and Quanrud \cite{CQ22} reduced updates to $\bigO{\log^2(n)/\eps^4}$ amortized time.

In this paper, we focus on solving a variant of DSG, the directed densest subgraph problem (DDSG). In a directed graph $G=(V,E)$, we think of two subsets $S,T\subseteq V$ as dense if there is a large number of edges originating in $S$ and ending in $T$. The standard definition for directed subgraph density was introduced in 1999 by Kannan and Vinay \cite{KV99}, and is defined as $\rho_G(S,T)=\frac{\sz{E(S,T)}}{\sqrt{\sz{S}\sz{T}}}$. The directed densest subgraph problem asks us to find $S,T\subseteq V$ such that the density is maximized.

Directed graphs generalize undirected graphs, and there are many applications interested in dense subgraphs of directed graphs rather than undirected graphs. Thus, DDSG has become a research topic of applied interest \cite{DBLP:journals/tkde/ChenS12, DBLP:conf/sigmod/MaFCLH22, DBLP:journals/sigmod/MaFCLZL21}. For example, in the context of social networks, directed dense subgraphs can model fake follower detection, where a user has a disproportionate number of followers in a certain set of users \cite{DBLP:conf/kdd/HooiSBSSF16}. In the context of the Internet, webpages can be classified into hubs (links to many pages) and authorities (many pages link to it), and directed dense subgraphs capture this relationship \cite{DBLP:journals/jacm/Kleinberg99}.

DDSG is polynomial-time solvable. Charikar \cite{Cha00} showed that there is an LP formulation for DDSG parametrized by a variable $c$, a guess for the ratio of $\sz{S}$ to $\sz{T}$. Solving the LP for all $\bigO{n^2}$ values of $c$, we may obtain an exact solution. Charikar also gave a $\bigO{n^2(n+m)}$ $1/2$-approximate algorithm, which follows from the LP formulation. In 2009, Khuller and Saha \cite{KS09} improved this bound to just $\bigO{n+m}$.

This work focuses on the dynamic variant of DDSG. Recent works have developed dynamic data structures that efficiently maintain approximations to the DDSG. In particular, \cite{SW20} was the first to maintain a $\epsless$-approximate directed densest subgraph in worst-case $\bigO{\log^5 (n)/\eps^7}$ time per edge update using edge orientation techniques. 

\subsection{Main results}

The main result of this work is a fully dynamic data structure for DDSG handling edge insertions and deletions with polylogarithmic amortized and worst-case update times and linear space. In this setting, we start with an empty graph for a fixed set of vertices $V$. We receive a stream of edges to be added into or deleted from the graph, one by one. After each edge update, the data structure is expected to maintain an estimate of the optimum subgraph density, and be able to list the vertices of an approximate densest subgraph efficiently.

\begin{thm}
Given a directed graph $G$ updated by edge insertions and deletions, there is an algorithm that maintains a $\apxless$-approximate directed densest subgraph in $\softO{\log^3 (n)/\eps^5}$ amortized time per update and $\softO{\log^4 (n)/\eps^7}$ worst-case time per update, where $\tilde{\mathcal{O}}$ hides $\log\log$ factors.
\end{thm}

This result improves slightly on previous bounds for approximate DDSG (\cite{SW20}, which finds $\bigO{\log^5(n)/\eps^7}$ time per update).

\subsection{Organization}

The rest of this document is organized as follows. In Section \ref{ch:overview}, we give a high-level overview of the techniques employed within this paper. In Section \ref{ch:preliminaries}, we formally define the problems and terminology, and explain the reduction from the directed densest subgraph problem to the vertex-weighted densest subgraph problem. In Section \ref{ch:lps}, we give LP relaxations of the densest subgraph problem, and show how this motivates approaching this problem using edge orientations. In Section \ref{ch:levels}, we give a data structure that maintains an approximate vertex-weighted densest subgraph. In Section \ref{ch:threshold}, we modify the data structure in the low-density case in order to obtain improved final bounds.

\section{High-level overview}\label{ch:overview}

In this section we give a high-level overview of the techniques employed in this paper.

\paragraph{DDSG reduces to VWDSG.}

As mentioned above, DDSG reduces to the vertex-weighted undirected densest subgraph problem (VWDSG). 

In VWDSG, the input consists of an undirected graph $G=(V,E)$ with vertex weights $w:V\to \Rgz$. For a subset $S\subseteq V$, let $\wt{S}=\sum_{v\in S} \wt{v}$ be the sum of all vertex weights in $S$.
The density of a set of vertices $S$ is defined as the number of edges in the subgraph induced by $S$ over the total vertex weight of $S$; i.e., $\rho_G(S)=\frac{\sz{E(S)}}{\wt{S}}$.

\paragraph{Approximate DDSG reduces to approximate VWDSG.}

\cite{Cha00} observed that $\epsmore$-approximate DDSG reduces to $\bigO{\log(n)/\eps}$ instances of VWDSG. Let $G$ be a directed graph. Recall that we define a subgraph as a pair of subsets $S, T\in V$, consisting of all edges directed from $S$ to $T$. In the reduction, we guess the sizes of sets $S$ and $T$ that produce the densest subgraph. For a parameter $t=\sqrt{\sz{S}/\sz{T}}$, we can instantiate an undirected vertex-weighted graph $G_t$, such that for a good choice of $t$, the optimal subgraph densities of $G$ and $G_t$ match. We can show that the optimum VWDSG of $G_t$, over all choices of $t$ equal to a power of $\epsmore$ in the range $[1/\sqrt{n}, \sqrt{n}]$, induces a $\epsmore$ approximate DDSG.

Thus, dynamic DDSG reduces to $\bigO{\log (n) / \eps}$ instances of dynamic VWDSG. For each instance of VWDSG we employ a data structure based on prior work \cite{SW20,CQ22} for the undirected and unweighted densest subgraph. We first explain the unweighted prior work to give context for our arguments.

\paragraph{The LP dual of DSG is edge orientation.}

The LP dual of DSG is an edge orientation problem, where the goal is to orient the edges to minimize the maximum in-degree \cite{Cha00}. There was prior work on dynamically maintaining such orientations approximately, and \cite{SW20} realized that one can extract approximate densest subgraphs from these orientations.

\paragraph{Locally optimal edge orientations imply globally optimal ones.}

The dynamic algorithms approximating these orientations are based on maintaining local optimality conditions. Our work is based on a local optimality condition from \cite{CQ22}, which for the unweighted setting is as follows: for $\alpha = \bigO{\eps^2 / \log(n)}$, for any edge $\{u,v\}$ oriented as $\arc{u}{v}$, 
\begin{align*}
    \indeg{v} \leq \alphamore \indeg{u} + \bigO{1}.
\end{align*}
Intuitively, we should not orient an edge $\{u,v\}$ towards $v$ if $v$ has in-degree substantially larger than $u$. \cite{CQ22} showed how to maintain this invariant efficiently, in polylogarithmic time per edge update.

Note the $\bigO{1}$-additive factor in the local optimality condition. This term is included for technical reasons to ensure some separation when the in-degrees are very small.
\cite{SW20} uses a different local optimality condition that only has a similar additive term.
An artifact of this additive term is that an orientation satisfying the invariant above leads to a bicriteria approximation for the densest subgraph. More precisely, one obtains a subgraph with density at least
\begin{align*}
    (1 - \epsilon) \OPT - \bigO{\frac{\log(n)}{\eps}}.
\end{align*}
We want a purely multiplicative approximation guarantee. \cite{SW20} observed that since $\OPT \geq 1$, one can address the additive term by simply duplicating every edge $\bigO{\log(n) / \eps^2}$ times, artificially making $\OPT \geq \bigO{\log(n) / \eps^2}$ and turning the additive error a relatively small factor.

\paragraph{Adapting to weighted edge orientations.}

It is natural to adapt this local optimality condition to VWDSG, as follows. For ease of notation, we assume the weights have been rescaled so that the minimum weight is $1$. (This also scales down the density.)
For any edge $\{u,v\}$ oriented as $\arc{u}{v}$, we have
\begin{align*}
    \load{v} \leq \alphamore \load{u} + \bigO{1}, &&
    \text{where } \load{v} = \frac{\indeg{v}}{\wt{v}}
\end{align*}
is defined to be the \emph{load} of vertex $v$. This notion of load is suggested by the LP.
Again, the vertex weights in the VWDSG instance must be at least $1$. Note that the addition or deletion of an arc changes the load of a vertex by $1/\wt{v}$, so rescaling the weights in this manner allows us to bound the effect of any one edge on the orientation.

With care, the techniques of \cite{CQ22} can be extended to maintain this local optimality condition with comparable polylogarithmic update times. Maintaining this invariant leads to a similar overall bicriteria approximation for VWDSG:
\begin{align*}
    \epsless \OPT - \bigO{\frac{\log(nW)}{\eps}},&&
    \text{ where } W = \max_v w(v).
\end{align*}
However, unlike unweighted DSG, $\OPT$ is not guaranteed to be at least $1$, and could be as small as $1/W$. Simply duplicating each edge $\bigO{\log(nW)/\eps^2}$ times does not increase the density by a large enough factor such that the additive term becomes relatively small.

\paragraph{Converting from bicriteria to purely multiplicative approximations.}

In the general case, it is not clear that we can remove this additive error for VWDSG. However, we are primarily interested in instances of VWDSG that arise from the DDSG reduction. Recall that in the reduction, the vertex weights are all of the form $w(v) \in \{t, 1/t\}$ for some parameter $t \in [1,\bigO{\sqrt{n}}]$ (before rescaling). Here we introduce our key lemma. We observe that for the optimal choice of $t$, the optimal density in the corresponding instance of VWDSG is at least $t$ before rescaling, or at least $1$ after rescaling. Thus, our strategy of duplicating edges $\bigO{\log(n)/\eps^2}$ times can indeed be applied in the same way as in DSG, and we obtain a purely multiplicative approximation guarantee for DDSG.

\paragraph{Optimizing the algorithm for low density.}

Finally, we are able to make a small modification to improve our data structure in the low-density regime. We note that optimizations can be made separately in the low-density and high-density cases. When $\OPT$ is high enough, edge duplication is not necessary, since the additive factor is small enough in comparison. On the other hand, the recursive depth of our subroutines is proportional to $\log(\OPT)$, so we obtain lower running times when $\OPT$ is small.

In order to utilize this low-density optimization, we introduce a threshold $T$, which limits the load of vertices to be at most $T$. This limits the recursive depth to $\log(T)$, although our estimates are now only valid when $\OPT<\epsless T$. We set $T=\bigTheta{\log^2(n)/\eps^4}$, which can handle the case where $\OPT=\bigO{\log(n)/\eps^2}$ with edges duplicated $\bigO{\log(n)/\eps^2}$ times. We run this data structure in parallel with our original data structure (now without duplicating edges), and for each subgraph query take the better of the two estimates. Overall, this gives slightly faster amortized and worst-case running times.

\paragraph{Summary.}

In conclusion, our main contributions are as follows. First, we generalize the \cite{CQ22} algorithm from DSG to VWDSG via the load-based local optimality invariant above. Incorporating the vertex weights into \cite{CQ22} adds a layer of complexity to the algorithm and the analysis.
Second, to resolve the discrepancy regarding the additive error, we provide the key lemma showing that the additive error scales nicely with $\OPT$ for the instances of VWDSG arising from DDSG. All put together, we obtain a dynamic approximation algorithm for DDSG with polylogarithmic update times.

\subsection{Comparison to \cite{SW20}}

\cite{SW20} also considered dynamic algorithms for DDSG and VWDSG. We had difficulty confirming the \cite{SW20} algorithm and analysis for DDSG and VWDSG, and here we try to explain our difficulties.

For context, we first review the main points of \cite{SW20} for undirected DSG at a high level. Similar to \cite{CQ22}, \cite{SW20} maintains an orientation satisfying a local optimality condition, except their optimality condition is expressed purely in additive terms. More precisely, for every edge $\{u,v\}$ directed as $\arc{u}{v}$, they have
\begin{align*}
    \indeg{v} \leq \indeg{u} + \eta
\end{align*}
for a parameter $\eta \geq 1$.

Ideally, one sets $\eta = \Theta({\eps^2 \OPT / \log(n)})$. Of course, $\OPT$ is not known a priori, so \cite{SW20} instantiates $\bigO{\log n}$ instances of the data structure for different values of $\eta$ corresponding to different guesses for $\OPT$. Furthermore, to ensure $\eta \geq 1$, \cite{SW20} also duplicates edges $\bigO{\log(n) / \epsilon^2}$ times
to ensure $\OPT = \bigOmega{1}$. At any point in time, the data structure corresponding to the best guess for $\OPT$ gives a $(1-\epsilon)$-approximate densest subgraph.

One can try to extend this approach to VWDSG. First, one rescales the vertex weights to be greater than or equal to $1$.
In the local optimality condition above, the in-degrees $\indeg{v}$ are replaced by loads $\ld{v}$.
However, the issue now is that $\OPT$ is not necessarily $\bigOmega{1}$, and an unbounded number of duplicates may be needed to raise $\eta$ to be $\geq 1$. \cite{SW20} seems to assume that $\OPT \geq 1$ for arbitrary VWDSG instances, and 
claims a $(1-\epsilon)$-approximation with polylogarithmic updates for VWDSG in general. We were not able to verify this result. However, we suspect that our key lemma mentioned above --- about $\OPT$ being sufficiently large in the critical VWDSG instance of the DDSG reduction --- can be applied here, as well.
We believe that this lemma should imply that \cite{SW20} gives a $(1-\epsilon)$ approximation to DDSG via VWDSG with polylogarithmic worst-case updates.

\section{Preliminaries}\label{ch:preliminaries}

We first formally define the definitions of the densest subgraph problem and the variations discussed in this work.

\subsection{The densest subgraph problem (DSG)}

In DSG, the input consists of an undirected graph $G=(V,E)$. The \ita{density of $S$}, $\rho_G(S)$, is defined as the density of the subgraph induced by $S$. That is, $\rho_G(S) = \frac{\sz{E(S)}}{\sz{S}}$, where $E(S)\subseteq E$ is the set of edges with both endpoints in $S$. The densest subgraph problem (DSG) asks us to find a set $S$ maximizing $\rho_G(S)$. We let $\densityopt(G)=\max_S \rho_G(S)$ denote the optimum density.

\subsection{The vertex-weighted densest subgraph problem (VWDSG)}\label{sec:defn-vwdsg}

In VWDSG, the input consists of an undirected graph $G=(V,E, w)$ with vertex weights $w:V\to \Rgz$. For a subset $S\subseteq V$, let $\wt{S}=\sum_{v\in S} \wt{v}$ be the sum of all vertex weights in $S$. The density is similarly defined as $\rho_G(S)=\frac{\sz{E(S)}}{\wt{S}}$, and the vertex-weighted densest subgraph problem (VWDSG) asks us to find a set $S$ maximizing $\rho_G(S)$. We let $\densityopt(G)=\max_S \rho_G(S)$ denote the optimum density.

Later, we will see that we prefer for vertex weights to be in $\Rgeo$. From the general case, we can transform the graph into the $\Rgeo$ case by multiplying all weights by $1/w_{\min}$; now, the minimum weight is $1$, and the maximum weight is $W=w_{\max}/w_{\min}$. The density of this normalized graph is multiplied by a factor of $w_{\min}$, so we divide our estimate by $w_{\min}$ afterwards.

\subsection{The directed densest subgraph problem (DDSG)}

The directed densest subgraph problem was first described by Kannan and Vinay \cite{KV99}. In DDSG, input consists of an unweighted directed graph $G=(V,E)$. For any two subsets $S, T\subseteq V$, let $E(S,T)$ be the set of all edges directed from a vertex in $S$ to a vertex in $T$. The density of these subsets $S,T$ is defined as $\rho_G(S,T)=\frac{\sz{E(S,T)}}{\sqrt{\sz{S} \sz{T}}}$, and the directed densest subgraph problem (DDSG) asks us to find sets $S, T$ maximizing $\rho_G(S,T)$. We let $\densityopt(G)=\max_{S,T} \rho_G(S,T)$ denote the optimum density.

\subsection{DDSG-VWDSG reduction}\label{sec:reduction}

There is a well-known reduction from DDSG to VWDSG \cite{Cha00, KS09, SW20}. We briefly go over the main ideas.

Given a directed graph $G=(V, E)$ and a parameter $t>0$, we can create an weighted undirected graph $G_t=(V_t, E_t, w_t)$, where
\begin{itemize}
    \item $V_t=V_t^{(L)}\union V_t^{(R)}$, where $V_t^{(L)}$ and $V_t^{(R)}$ are copies of $V$;
    \item $E_t=\{(u,v)\vert u\in V_t^{(L)}, v\in V_t^{(R)}, (u, v)\in E\}$ 
    \item $w_t(v)=\begin{cases}1/2t& v\in V_t^{(L)}\\t/2& v\in V_t^{(R)}\end{cases}$.
\end{itemize}
\begin{figure}[t]
\centering
\begin{minipage}[c]{0.85\textwidth}
\begin{framed}
    \includegraphics[width=0.43\textwidth]{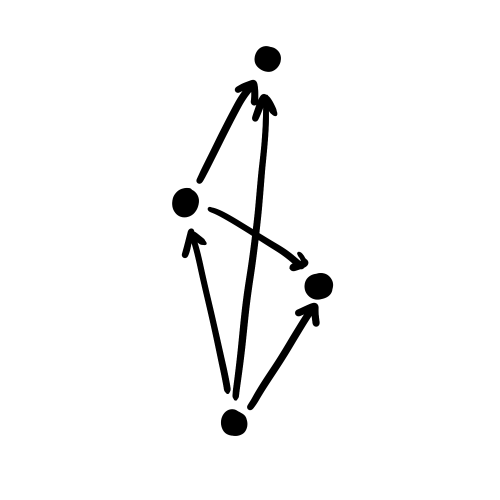}
    \hspace{1em}
    \rulesep
    \hspace{1em}
    \includegraphics[width=0.43\textwidth]{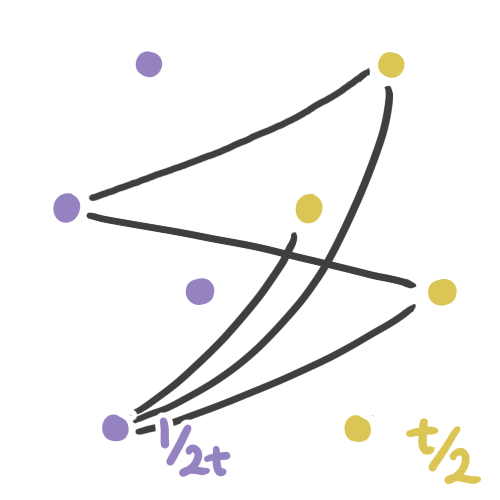}
\end{framed}
\end{minipage}
    \caption{An example reduction transforming a directed graph $G$ on the left to the bipartite undirected graph $G_t$ on the right.}
    \label{fig:reductionexample}
\end{figure}

Consider sets $S, T\subseteq V$. In the directed graph, we have that $\rho_G(S,T)=\frac{\sz{E(S,T)}}{\sqrt{\sz{S}\sz{T}}}$. If we take $S^{(L)}\subseteq V^{(L)}$ corresponding to $S$, and $T^{(R)}\subseteq V^{(R)}$ corresponding to $T$, then $\rho_{G_t}\left(S^{(L)}\union T^{(R)}\right) = \frac{\sz{E(S,T)}}{\left((1/t)\sz{S}+t\sz{T}\right)/2}$.

We first show that the optimum subgraph density of $G_t$ underestimates the optimum subgraph density of $G$ for any arbitrary value of $t$. The following lemma is shown in \cite{SW20}, and is implicit in \cite{Charikar00}. We include the proof here for completeness.

\begin{lemma}\label{lemma:reductionleq}
    Given a directed graph $G=(V,E)$, for any choice of $t$,
    \[\densityopt(G_t)\leq \densityopt(G).\]
\end{lemma}
\begin{proof}
    Let $S^{(L)}\union T^{(R)}$ be the densest subgraph in $G_t$. Then
    \begin{equation*}
        \densityopt(G_t)=\rho_{G_t}\left(S^{(L)}\union T^{(R)}\right) = \frac{\sz{E(S,T)}}{\left((1/t)\sz{S}+t\sz{T}\right)/2}\tago{\leq} \frac{\sz{E(S,T)}}{\sqrt{\sz{S}\sz{T}}}\leq \densityopt(G)
    \end{equation*}
    where \tagr is by the AM-GM inequality.
\end{proof}

We have now shown that each $G_t$ underestimates the optimum subgraph density of $G$. The next lemma shows that there exists some value of $t$ such that $\densityopt(G) = \densityopt(G_t)$ exactly.

\begin{lemma}\label{lemma:reductioneq}
    Given a directed graph $G=(V,E)$ and sets $S,T\subseteq V$ such that $\densityopt(G)=\rho_G(S,T)$, if $t=\sqrt{\frac{\sz{S}}{\sz{T}}}$, then
    \[\densityopt(G_t)=\densityopt(G).\]
\end{lemma}
\begin{proof}
    Consider set $S^{(L)}\subseteq V_t^{(L)}$ corresponding to $S$ and set $T^{(R)}\subseteq V_t^{(T)}$ corresponding to $T$. Then the density of $S^{(L)}\union T^{(R)}$ is upper bounded by $\densityopt(G_t)$.
    \begin{align*}
        \densityopt(G_t)\geq \rho_{G_t}\left(S^{(L)}\union T^{(R)}\right)&=\frac{\sz{E(S,T)}}{\left((1/t)\sz{S}+t\sz{T}\right)/2}\\
        &=\frac{\sz{E(S,T)}}{\left(\sqrt{\frac{\sz{T}}{\sz{S}}}\sz{S}+\sqrt{\frac{\sz{S}}{\sz{T}}}\sz{T}\right)/2}=\frac{\sz{E(S,T)}}{\sqrt{\sz{S}\sz{T}}}= \densityopt(G).
    \end{align*}

    Combining this with Lemma \ref{lemma:reductionleq} gives us $\densityopt(G_t)=\densityopt(G)$.
\end{proof}

While there is always some value of $t$ such that $G_t$ gives the exact optimum subgraph density of $G$, the problem now is that  we do not know this optimum value. However, since the optimum value of $t=\sqrt{\frac{\sz{S}}{\sz{T}}}$, and $\sz{S}, \sz{T}\leq n$, there are only $\bigO{n^2}$ possible values of $t$. Thus, it is possible to calculate $\densityopt(G_t)$ by trying all $\bigO{n^2}$ choices of $t$, and the maximum of these values gives $\densityopt(G)$.

\subsubsection{$\epsmore$-approximate reduction}

Creating $\bigO{n^2}$ auxiliary graphs $G_t$ is too expensive for our desired bounds. Fortunately, we don't have to run $\bigO{n^2}$ copies of our algorithm if we only want an approximate answer. Intuition says that a close enough estimate of $t$ will give us a good density estimate. Approximating the density by approximating some parameter $t$ was first noticed by \cite{Cha00}, and appears in various later works such as \cite{KS09, SW20, CQ22}.

\begin{lemma}\label{lemma:reductionapx}\cite{SW20}
    Given directed graph $G=(V,E)$ and sets $S,T\subseteq V$ such that $\densityopt(G)=\rho_G(S,T)$, if $\sqrt{\frac{\sz{S}}{\sz{T}}}\cdot \epsless\leq t\leq \sqrt{\frac{\sz{S}}{\sz{T}}}\cdot \epsmore$, then
    \[\densityopt(G_t)\geq \epsless\densityopt(G).\]
\end{lemma}
\begin{proof}
    Consider set $S^{(L)}\subseteq V_t^{(L)}$ corresponding to $S$ and set $T^{(R)}\subseteq V_t^{(T)}$ corresponding to $T$. Similar to Lemma \ref{lemma:reductioneq}, we have
    \[\densityopt(G_t)\geq \rho_{G_t}\left(S^{(L)}\union T^{(R)}\right)=\frac{\sz{E(S,T)}}{\left((1/t)\sz{S}+t\sz{T}\right)/2}.\]
    Then
    \[
        \densityopt(G_t)\geq\frac{\sz{E(S,T)}}{\left(\frac{1}{\epsless}\sqrt{\frac{\sz{T}}{\sz{S}}}\sz{S}+\epsmore\sqrt{\frac{\sz{S}}{\sz{T}}}\sz{T}\right)/2}=\epsless\frac{\sz{E(S,T)}}{\sqrt{\sz{S}\sz{T}}}= \epsless\densityopt(G).
    \]
\end{proof}

\subsubsection{Using the reduction}\label{sec:usereduction}

We now show how to apply dynamic algorithms for the VWDSG problem to obtain a dynamic algorithm for DDSG. First, we show a lemma giving a lower bound on the optimum subgraph density of a directed graph. The motivation for this lemma comes from the fact that our data structures for VWDSG give a bicriteria approximation guarantee, with both a multiplicative and an additive error. Recall that we seek a purely multiplicative approximation. To drop the additive factor, we must show that the optimum subgraph density is much larger relative to the additive error. The following lemma, while simple in hindsight, is critical for this argument.

\begin{lemma}\label{lemma:geqt}
    Given a directed graph $G=(V,E)$, suppose we know $S,T\subseteq V$ such that $\rho_G(S,T)=\densityopt(G)$. Then $\densityopt(G)\geq \max\left(\sqrt{\frac{\sz{S}}{\sz{T}}}, \sqrt{\frac{\sz{T}}{\sz{S}}}\right)$.
\end{lemma}
\begin{proof}
    First, note that
    \[
        \densityopt(G)=\rho_G(S,T)=\frac{\sz{E(S,T)}}{\sqrt{\sz{S}\sz{T}}}.
    \]
    
    We claim every vertex $v\in S$ has at least one outgoing edge ending in $T$. If not, then we can remove $v$ from $S$ at no penalty, which means $\rho_G(S\setminus\{v\}, T)=\frac{\sz{E(S, T)}}{\sqrt{(\sz{S}-1)\sz{T}}}\geq \rho_G(S,T)=\densityopt(G)$, which is a contradiction.

    Thus, $\sz{E(S, T)}\geq \sz{S}$, and
    \[
        \densityopt(G)=\frac{\sz{E(S,T)}}{\sqrt{\sz{S}\sz{T}}}\geq \frac{\sz{S}}{\sqrt{\sz{S}\sz{T}}}=\sqrt{\frac{\sz{S}}{\sz{T}}}.
    \]

    Similarly, we claim every vertex $u\in T$ has at least one incoming edge edge starting in $S$. If not, then we can remove $u$ from $T$ at no penalty, which means $\rho_G(S, T\setminus\{u\})=\frac{\sz{E(S, T)}}{\sqrt{\sz{S}(\sz{T}-1)}}\geq \rho_G(S,T)=\densityopt(G)$, which is a contradiction.

    Thus, $\sz{E(S, T)}\geq \sz{T}$, and
    \[
        \densityopt(G)=\frac{\sz{E(S,T)}}{\sqrt{\sz{S}\sz{T}}}\geq \frac{\sz{T}}{\sqrt{\sz{S}\sz{T}}}=\sqrt{\frac{\sz{T}}{\sz{S}}}.
    \]
\end{proof}

Finally, we describe the reduction from a dynamic $\epsmore$-approximate DDSG to a dynamic, bicriteria $(1 + \eps, \eps)$-approximation for VWDSG.
    
\begin{thm}\label{thm:ddsg}
    We are given a $n$-vertex directed graph $G$ dynamically updated by edge insertions and deletions, and we assume there exists an algorithm on vertex-weighted graphs that maintains a subgraph with density at most $\epsless\densityopt(G)-\bigO{\eps}$ in $T(n, W, \eps)$ time per update, where $W$ is the maximum vertex weight. Then there exists a $\apxless$-approximate algorithm for the directed densest subgraph problem in $T(n, n, \eps)\cdot \bigO{\log (n)/\eps}$ time per update.
\end{thm}
\begin{proof}
    Let $\mathcal{T}=\left\{\frac{1}{\sqrt{n}},(1+\eps)\frac{1}{\sqrt{n}},(1+\eps)^2\frac{1}{\sqrt{n}},\hdots,\sqrt{n}\right\}$. We run a parallel copy of the VWDSG algorithm for each $t\in \mathcal{T}$.
    
    For each $t$, create $G_t$ as described in the reduction in Section \ref{sec:reduction}. To normalize vertex weights, divide the weight of all vertices by ${w_t}_{\min}=\min(1/2t, t/2)$. This decreases the density by a factor of $1/{w_t}_{\min}=\max(2t, 2/t)$, so later we multiply any density estimates obtained by $1/{w_t}_{\min}$. (Note that scaling weights uniformly does not affect the nodes that comprise a densest subgraph.) Then, we use our data structure to maintain a $\epsless$-approximate VWDSG in $G_t/{w_t}_{\min}$. 
    
    To query the maximum subgraph density, for each $t$, find $\rho_t\geq \epsless \densityopt(G_t)-\bigO{\eps}$, and output $\rho=\max_t \rho_t\cdot (1/{w_t}_{\min})$.

    Let $S,T\subseteq V$ produce the optimum subgraph density in $G$. Consider the copy of the data structure corresponding to the parameter $t$ such that $\sqrt{\frac{\sz{S}}{\sz{T}}}\cdot \epsless\leq t\leq \sqrt{\frac{\sz{S}}{\sz{T}}}\cdot \epsmore$. We have
    \begin{align*}
        \rho&\geq \rho_t\cdot (1/{w_t}_{\min})\\
        &\geq \epsless\densityopt(G_t)-\bigO{\eps}\cdot (1/{w_t}_{\min})\\
        &\tago{\geq} \epsless^2\densityopt(G)-\bigO{\eps}\cdot (1/{w_t}_{\min})\\
        &\geq \epsless^2\densityopt(G)-\bigO{\eps}\cdot 2\epsless \max(\sqrt{\sz{S}/\sz{T}}, \sqrt{\sz{T}/\sz{S}})\\
        &\tago{\geq} \epsless^2\densityopt(G)-\bigO{\eps}\cdot 2\epsless \densityopt(G)\\
        &=\apxless \densityopt(G),
    \end{align*}

    where \tagr is by Lemma \ref{lemma:reductionapx}, and \tagr is by Lemma \ref{lemma:geqt}.
\end{proof}

\section{LP relaxations and local optimality}\label{ch:lps}

We now introduce LP relaxations for DSG and VWDSG. The duals of these LPs let us interpret the densest subgraph optimization as an edge orientation problem, which motivates the data structure discussed afterwards.

\subsection{The LP relaxation for DSG}

\begin{figure}[t]
\begin{framed}
\begin{minipage}[t]{0.4\textwidth}
    \begin{align*}
        \max &\sum_{(u,v)\in E} \min(x_u, x_v)\\
        \text{over } &x_v\geq 0 \text{ for all } v\in V\\
        \sta & \sum_{v\in V}x_v\leq 1.
    \end{align*}
\end{minipage}
\begin{minipage}[t]{0.5\textwidth}
    \begin{align*}
        \min &\max_v \sum_{e\in \alldeg{v}} y(e,v)\\
        \text{over } &y(e,v)\geq 0 \text{ for all } e\in E\text{ and } v\in e\\
        \sta & y(e,u)+y(e,v)\geq 1 \text{ for all }e=(u,v)\in E.
    \end{align*}
\end{minipage}
\end{framed}
    \caption{The LP formulation and its dual for DSG.}
    \label{fig:dsglp}
\end{figure}

\cite{Cha00} first discovered an LP formulation for the densest subgraph problem. Recall that the goal of DSG is to find a $S\subseteq V$ of maximum density in $G$.

One way to model this problem is to set indicator variables $x_v\in \zo$ for each $v\in V$, where $x_v=1$ if $v\in S$, and $x_v=0$ otherwise. The density of $S$ can thus be expressed as $\rho_G(S)=\frac{\sz{E(S)}}{\sz{S}} = \frac{\sum_{(u,v)\in E} \min(x_u, x_v)}{\sum_{v\in V} x_v}$.

We relax this objective into an LP (see Figure \ref{fig:dsglp}). We allow $x_v\in \Rgez$, and normalize the denominator with the constraint $\sum_{v\in V}x_v=1$. Then, to rewrite $\sum_{(u,v)\in E} \min(x_u, x_v)$ as a linear objective, we can add constraints $z_e\leq x_u$ and $z_e\leq x_v$ for each $e=(u,v)\in E$, and sum $\sum_{e\in E} z_e$ instead. However, we keep the original formulation for simplicity.

Charikar \cite{Cha00} showed that this LP relaxation is exact, that is, $\lpopt = \densityopt(G)$.

\begin{lemma}\label{lemma:dsg-duality}
    \(\lpopt=\densityopt(G).\)
\end{lemma}

\subsection{The dual LP and integral orientations}

We can interpret the dual of this LP as taking each edge $e=(u,v)$ and partially directing it towards its endpoints with respect to the weights $y(e,u)$ and $y(e,v)$. Following \cite{SW20}, we define the \ita{load} on a vertex $v$ as $\ld{v}=\sum_{e\in \alldeg{v}} y(e,v)$. The objective of the dual is thus to minimize the maximum load of any node $v$. We call $y$ a \ita{fractional orientation}.

Note that if we restrict each $y(e,v)$ to $\zo$, we can interpret this as taking each (undirected) edge, and orienting it towards one of its endpoints as a directed edge. We call this an \ita{integral orientation}; in this case, the load is equivalent to the in-degree, which we denote by $\indeg{v}$.

The approach we take to solve DSG is based on maintaining an integral orientation of the input graph.

\subsection{The LP relaxation for VWDSG}

\begin{figure}[t]
\begin{framed}
\begin{minipage}[t]{0.4\textwidth}
    \begin{align*}
        \max &\sum_{(u,v)\in E} \min(x_u, x_v)\\
        \text{over } &x_v\geq 0 \text{ for all } v\in V\\
        \sta & \sum_{v\in V}\wt{v}\cdot x_v\leq 1.
    \end{align*}
\end{minipage}
\begin{minipage}[t]{0.5\textwidth}
    \begin{align*}
        \min &\max_v \frac{1}{\wt{v}} \sum_{e\in \alldeg{v}} y(e,v)\\
        \text{over } &y(e,v)\geq 0 \text{ for all } e\in E\text{ and } v\in e\\
        \sta & y(e,u)+y(e,v)\geq 1 \text{ for all }e=(u,v)\in E.
    \end{align*}
\end{minipage}
\end{framed}
    \caption{The LP formulation and its dual for VWDSG.}
    \label{fig:wdsglp}
\end{figure}

We write a similar LP for VWDSG \cite{SW20} as we do for DSG, which we describe in Figure \ref{fig:wdsglp}.

For the primal, note that the density can be written as $\rho_G(S)=\frac{\sz{E(S)}}{\wt{S}}=\frac{\sum_{(u,v)\in E}\min(x_u, x_v)}{\sum_{v\in V} \wt{v}\cdot x_v}$. The only difference is the denominator, which we normalized to $1$. Thus, only that constraint changes.

Taking the dual, we find that the load is redefined as $\ld{v}=\frac{1}{\wt{v}}\sum_{e\in \alldeg{v}} y(e,v)$. We use this definition for $\ld{v}$ in the rest of the paper. We can consider the load to be equivalent to the weighted in-degree in the case of an integral orientation.

A similar argument can be made as in \cite{Cha00} to show that this LP relaxation is also exact for the VWDSG case.

\begin{lemma}\label{lemma:vwdsg-duality}
    \(\lpopt=\densityopt(G).\)
\end{lemma}

Later on, we use this lemma implicitly.

\subsection{Local optimality}

Consider the following condition:
\begin{center}
    For all edges $e=(u,v)\in E$, if $y(e,v)>0$, then $\ld{v}\leq \ld{u}$.
\end{center}
Intuitively, this says that to minimize the maximum in-degree, we should never direct an edge towards an endpoint with higher load. If this condition is satisfied for all edges, we say that $y$ is \ita{locally optimal}.

We state, but do not prove, a lemma showing the power of this local optimality condition.

\begin{lemma}
    There is an optimum solution to the dual LP satisfying the local optimality condition. Conversely, if $y$ satisfies the local optimality condition, then $y$ is an optimum solution.
\end{lemma}

Computationally, it is easier to work with an approximation of the local optimality condition. \cite{CQ22} provides a bicriteria approximation guarantee that we adapt to the vertex-weighted case. For some $\alpha,\beta \geq 0$, we say that $y$ is \ita{$(\alpha,\beta)$-locally optimal}, or a \ita{local $(\alpha,\beta)$-approximation}, if
\begin{center}
    For all edges $e=(u,v)\in E$, if $y(e,v)>0$, then $\ld{v}\leq \alphamore\ld{u}+\beta$.
\end{center}
Intuitively, this says that we should never direct an edge towards an endpoint with a \ita{much} higher load.

We show that approximate local optimality guarantees approximate global optimality. The proof follows as in \cite{CQ22}, but we include it here for completeness.

\begin{lemma}\label{lemma:localopt}
    Let $\alpha>0$ with $\alpha<c/\log(nW)$ for sufficiently small constant $c$. Let $\mu=\max_v\ld{v}$. Suppose every arc is $(\alpha, \beta)$-locally optimal. Then
    \[\mu\leq e^{\bigO{\alpha\log(nW)}}\left(\lpopt+\bigO{\sqrt{\frac{\log(nW)}{\alpha}}}\beta\right).\]
    In particular, given $\eps\in (0,1)$ and sufficiently small constant $c$, if every arc is $(c\eps^2/\log(nW), \bigO{1})$-locally optimal, then $\mu\leq \epsmore\lpopt+\bigO{\log(nW)/\eps}$.
\end{lemma}

\begin{proof}
    We define an increasing sequence of \ita{levels} $\mu_0<\mu_1<\mu_2<\cdots$, where we let $\mu_0=0$ and $\mu_i=\alphamore\mu_{i-1}+\beta$.

    Let $k$ be the unique index such that $\mu_{k-1}\leq\mu<\mu_{k}$. Note that
    \[
        \mu_{k-i}\geq \frac{\mu}{\alphamore^i}-i\beta\geq e^{-\alpha i}\mu-i\beta.
    \]

    For every $i\in \{1,\hdots,k\}$, let $S_i=\set*{v}{\ld{v}\geq \mu_{k-i}}$. Note that for every $i\leq k$, $S_i$ is nonempty (the maximum load is $\mu<\mu_k$). Since $\wt{S_i}\leq nW$ for all $i$, for some sufficiently small parameter $\eps>0$ (which we will choose later), there must be an index $i^\star\leq \bigO{\log_{\epsmore}(nW)}=\bigO{\log(nW)/\eps}$ such that $\wt{S_{i^\star+1}}\leq \epsmore\wt{S_{i^\star}}$.

    Consider the subgraph induced by $S_{i^\star+1}$.  Vertices in $S_{i^\star}$ have load at least $e^{-\alpha i^\star}\mu-i^\star\beta \geq e^{-\bigO{\alpha \log(nW)/\eps}}\mu-\bigO{\log(nW)/\eps}\beta$. Furthermore, $\bigcup_{v\in S_{i^\star}}\incut{v}\subseteq E(S_{i^\star+1})$ because every edge satisfies $(\alpha, \beta)$-local optimality. Thus,
    \begin{align*}
        \lpopt &\geq \frac{\sz{E(S_{i^\star+1})}}{\wt{S_{i^\star+1}}}\\
        &\geq \frac{\sz{\bigcup_{v\in S_{i^\star}}\incut{v}}}{\wt{S_{i^\star+1}}} = \frac{\sum_{v\in S_{i^\star}}\ld{v}\cdot \wt{v}}{\wt{S_{i^\star+1}}}\\
        &\geq \left(e^{-\bigO{\alpha \log(nW)/\eps}}\mu-\bigO{\log(nW)/\eps}\beta\right)\frac{\wt{S_{i^\star}}}{\wt{S_{i^\star+1}}}\\
        &\geq \frac{1}{1+\eps}\left(e^{-\bigO{\alpha \log(nW)/\eps}}\mu-\bigO{\log(nW)/\eps}\beta\right).
    \end{align*}
    Rearranging, we get
    \begin{align*}
        \mu &\leq e^{\bigO{\alpha \log(nW)/\eps}}\left(\epsmore\lpopt+\bigO{\log(nW)/\eps}\beta\right)\\
        &\leq e^{\bigO{\alpha \log(nW)/\eps} + \eps}\left(\lpopt+\bigO{\log(nW)/\eps}\beta\right)\\
        &\leq e^{\bigO{\sqrt{\alpha \log(nW)}}}\left(\lpopt+\bigO{\sqrt{\log(nW)/\alpha}}\beta\right)
    \end{align*}
    when we set $\eps=\bigO{\sqrt{\alpha \log(nW)}}$.
\end{proof}

This proof also shows that it is easy to obtain an approximate densest subgraph from $\ell$. Consider the vertices sorted in decreasing order of load. The proof shows that some prefix of this list is an approximate densest subgraph. To identify such a prefix, we only need to keep track of the cardinalities of the sets $S_i$ as defined in the proof, and finding $i^\star$ such that $\sz{S_{i^\star+1}}\leq \epsmore\sz{S_{i^\star}}$.

\subsection{Prior work for DSG, and intuition}\label{sec:intuition}

We would like to maintain a locally optimal orientation. Consider the layers $S_i$ from the proof of Lemma \ref{lemma:localopt}. Each $S_i$ contains nodes a maximum of a $\alphamore$-multiplicative factor of load apart. The (exact) local optimality condition implies that an edge can never be oriented from a lower layer to a higher layer. Our approximate local optimality condition allows for a little bit more slack; an edge is also allowed to point up a constant number of layers from $S_i$ to $S_{i-1}$.

Our data structure roughly sorts the nodes into these layers to help maintain approximate local optimality. When processing an edge insertion, the load on some node goes up, and we keep track of the movement of that node between layers. Ideally, the node only goes up by at most one layer at a time. This would allow us to guarantee that adding edges doesn't (momentarily) invalidate the local optimality conditions by too much, especially when we start manipulating edges to balance the load.

This work is inspired by \cite{CQ22}, which describes a similar local optimality condition, and show a dynamic data structure for maintaining a $\epsless$-approximate densest subgraph for unweighted graphs in amortized $\bigO{\log^2(n)/\eps^4}$ time per update. We are able to apply many of these ideas to the vertex-weighted case, with the vertex weights $1$ or greater. Intuitively, since the load on a vertex $v$ increases in increments of $1/\wt{v}$, if the data structure can handle loads increasing by $1$, then it can certainly handle loads increasing by less.

We describe the data structure formally below.

\section{Vertex-weighted densest subgraph}\label{ch:levels}

Formally, the setup is as follows. We are initially given an empty weighted graph with $n$ vertices and weights between $1\leq \wt{v}\leq W$. The data structure processes a series of edge insertions, and maintains a integral orientation of the edges (with some modifications). The following theorem states the guarantees of this data structure:

\begin{thm}\label{thm:vwdsg}
    Consider the problem of approximating the densest subgraph in a vertex-weighted graph with weights $w:V\to \Rgeo$, with the graph dynamically updated by edge insertions and deletions. Let $\eps>0$ be given. We can maintain an orientation (explicity) with maximum in-degree at most $\epsmore\lpopt + \bigO{\log(nW)/\eps}$, and a subgraph (implicitly) with density at least $\epsless\lpopt - \bigO{\log(nW)/\eps}$, in $\bigO{\log^2(nW)\log(\log(nW)/\eps)/\eps^2}$ amortized time per update and $\bigO{\log^3(nW)\log(\log(nW)/\eps)/\eps^4}$ worst-case time per update, where $W=\max_v \wt{v}$.
\end{thm}

\subsection{The data structure}

The data structure described below maintains an integral orientation of a vertex-weighted undirected graph. It also maintains a list of the layers $S_i$ (as described in the proof of Lemma \ref{lemma:localopt}) in order of decreasing load, as well as the cardinalities of said sets, so that it is simple to find the vertices that comprise a densest subgraph. We refer to this as an \ita{implicit} representation of the approximate densest subgraph; in particular, we can list off the vertices in an approximate densest subgraph in $\bigO{1}$ time per vertex. By the construction of our data structure, it is easy to see that maintaining these lists is simple, and that doing so takes negligible time. (In particular, at each step a vertex only goes up or down one level, so updating the sets $S_i$ is easy.) We omit the maintenance of these lists from the pseudocode, as they distract from the main ideas. We discuss the implementation details later. 

\subsubsection{Data structure overview}

Let $\alpha=c\eps^2/\log(nW)$, for a sufficiently small constant $c$. The data structure maintains a $(\bigO{\alpha},\bigO{1})$-locally optimal orientation as new edges are added to or deleted from the graph. The basic idea is to flip the orientation of any edge that does not satisfy the local optimality conditions.

We run into three main concerns when implementing our data structure. First, it must be able to efficiently detect arcs that violate local optimality. Second, it must bound the effect of ``cascades'', where flipping one arc causes more local optimality violations, and thus even more arc flips. Third, it must ensure that any one arc flip does not increase the load by too much.

\subsubsection{Arc labels}

For each arc $a=\arc{u}{v}$, we define a label on the arc $\lb{a}$ that records a snapshot of $\ld{v}$ whenever we process that arc. The data structure uses these labels to maintain approximate local optimality. The labels are also useful for the running time analysis later.

The data structure makes decisions about the orientation based on $\lab{\arc{u}{v}}$ rather than $\ld{v}$ because $\lab{\arc{u}{v}}$ is updated much less frequently. If each arc kept track of $\ld{v}$ directly, a small change to $\ld{v}$ would need to be propagated to many incoming arcs, which quickly blows up as $\ld{v}$ increases. On the other hand, we only need to make sure our labels $\lab{\arc{u}{v}}$ stay relatively close to the true value of $\ld{v}$, updating only when the difference is too great. The delayed updates to $\lab{\arc{u}{v}}$ make the data structure efficient while still retaining approximate local optimality.

\subsubsection{Levels}

We define a series of levels $L_i$ such that $L_0=0$ and $L_i=\alphamore L_{i-1}+1$. For some value $x$, we say that $x$ is at level $i$ if $L_{i-1}<x\leq L_i$. We define an auxiliary level function $\mathcal{L}$ such that $\lv{x}=i$ if $x$ is at level $i$. We use this level function to organize the arcs in our data structure.

We can make $\lv{x}$ an explicit function with a few elementary calculations.
Observe that
\begin{align*}
L_i=\sum_{j=0}^{i-1}\alphamore^j&=\frac{\alphamore^i-1}{\alpha}.\\
\intertext{Therefore,}
\alpha L_i + 1 &=\alphamore^i\\
\log_{\alphamore}(\alpha L_i + 1)&=i
\end{align*}
and so
$\lv{x}=\lceil\log_{\alphamore}(\alpha x + 1)\rceil$.

Recall levels $\mu_i$ used in the proof of Lemma \ref{lemma:localopt}; $L_i$ are exactly these levels, except we fix $\beta=1$.

Organizing the arcs into similar-load levels like so allows us to maintain labels slightly more efficiently for outgoing arcs, since now we can sort labels of outgoing arcs by levels rather than exact label values. Furthermore, we will show that levels still encapsulate the lower-level bicriteria approximations. This allows us to compare loads and labels more easily and simplify our calculations.

We introduce some lemmas about levels.
The first lemma maps inequalities between levels to inequalities between the underlying values. The proofs are straightforward and omitted.

\begin{lemma}
For our level function $\lv{x}$, the following statements hold:
    \begin{itemize}
    \item If $x\leq \alphamore y+1$, then $\lv{x}\leq \lv{y}+1$,

    \item If $\lv{x}\leq \lv{y}$, then $x\leq \alphamore y+1$,

    \item If $x\geq \alphamore y+1$, then $\lv{x}\geq \lv{y}+1$,

    \item If $\lv{x}\geq \lv{y}+2$, then $x\geq \alphamore y+1$.
    \end{itemize}
\end{lemma}

When we insert or delete an arc directed into a node $v$, $\ld{v}$ goes up or down by $1/\wt{v} \leq 1$.
By defining $L_i$ so that $L_i \geq L_{i-1} + 1$, we ensure that $\lv{\ld{v}}$ also increases or decreases by no more than 1 per edge update.

\begin{lemma}\label{lemma:lv-smooth}
For our level function $\lv{x}$, the following statements hold:
    \begin{itemize}
    \item $\lv{x+1}\leq \lv{x}+1$,
    \item $\lv{x-1}\geq \lv{x}-1$.
    \end{itemize}
In particular, a single edge insertion, deletion, or flip only changes the level of its endpoints by at most one. i.e.,
    \begin{itemize}
        \item $\lv{\frac{\indeg{v} + 1}{\wt{v}}}= \lv{\ld{v} + 1/\wt{v}}\leq \lv{\ld{v}} + 1$,
        \item $\lv{\frac{\indeg{v} - 1}{\wt{v}}}= \lv{\ld{v} - 1/\wt{v}}\geq \lv{\ld{v}} - 1$.
    \end{itemize}
\end{lemma}

For an arc $a=\arc{u}{v}$, if $\lv{\ld{v}}\leq \lv{\ld{u}}$, then $\ld{v}\leq \alphamore\ld{u}+1$. Thus, maintaining the the invariant $\lv{\ld{v}}\leq \lv{\ld{u}}$ for all arcs ensures $(\alpha, 1)$-local optimality.
More generally, if all arcs are oriented such that $\lv{\ld{v}}\leq \lv{\ld{u}}+\bigO{1}$, then the orientation is still $(\bigO{\alpha}, \bigO{1})$-locally optimal, per the following lemma.

\begin{lemma}
    Suppose for all arcs $a=\arc{u}{v}$,  $\lv{\ld{v}}\leq \lv{\ld{u}}+\bigO{1}$. Then the orientation is $(\bigO{\alpha}, \bigO{1})$-locally optimal.
\end{lemma}

\begin{proof}
    It suffice to show that if $\lv{x} \leq \lv{y} + 1$, then
    \begin{align*}
        x \leq (1 + C \alpha)y + D
    \end{align*}
    for fixed constants $C, D > 0$. (For more than one level, up to a constant, we can iterate on this inequality.)
    We have
    \begin{align*}
        x &\leq %
        \frac{(1 + \alpha)^{\lv{x}} - 1}{\alpha} %
        \leq %
        \frac{(1 + \alpha)^{\lv{y}+1} - 1}{\alpha}
        \\
        &=
        (1+\alpha)^2 \frac{(1 + \alpha)^{\lv{y}-1} - 1}{\alpha} + \frac{(1+\alpha)^2 - 1}{\alpha}
        \leq
        (1+\alpha)^2 y + 2 + \alpha
    \end{align*}
    as desired.
\end{proof}

\subsubsection{The algorithm}

Let $C$ be a sufficiently large constant.

\begin{algorithm}[p]
\caption{\code{insert}($e=(u,v)$)}
\label{alg:insert3}
    Let $\ld{u}\geq \ld{v}$. (Otherwise, swap $u$ and $v$.)\\
    Orient $e$ as $a=\arc{u}{v}$ and add it to the orientation, then set $\lb{a}=\ld{v}$.\\
    Call \code{check-inc}($v$).
\end{algorithm}

\begin{algorithm}[p]
\caption{\code{delete}($e=(u,v)$)}
\label{alg:delete3}
    Let $e$ be oriented as $a=\arc{u}{v}$; delete $a$ from the orientation.\\
    Call \code{check-dec}($v$).
\end{algorithm}

\begin{algorithm}[p]
\caption{\code{check-inc}($v$)}
\label{alg:checkinc3}
\tcc{This subroutine is called whenever $\ld{v}$ has increased.}
\For {up to $C/\alpha$ arcs $a=\arc{u}{v}\in \incut{v}$, in increasing order of $\lb{a}$}
{
    \tcc{The label is much lower than the current load.}
    \uIf {$\lv{\ld{v}}\geq\lv{\lb{a}}+2$}
    {
        \tcc{The arc is too imbalanced.}
        \uIf {$\lv{\ld{u}}+2\leq\lv{\ld{v}}$} 
        {
            Flip $a$ to $\arc{v}{u}$, then set $\lb{a}=\ld{u}$.\\
            \tcc{This restores the original value of $\ld{v}$, but now $\ld{u}$ has increased.}
            Call \code{check-inc}($u$).
            Return.
        }
        \Else
        {
            Set $\lb{a}=\ld{v}$.
        }
    }
    \tcc{There are no arcs that satisfy the above conditions; all arc labels are close to the current load.}
    \Else
    {
        Return.
    }
}
\end{algorithm}

\begin{algorithm}[p]
\caption{\code{check-dec}($u$)}
\label{alg:checkdec3}
\tcc{This subroutine is called whenever $\ld{u}$ has decreased.}
Choose an arc $a=\arc{u}{v}\in\outcut{u}$ such that $\lv{\lb{a}}$ is maximized.\\
\If {$\lv{\ld{u}}+3\leq \lv{\lb{a}}$}
{
    \tcc{The arc is too imbalanced.}
    Flip $a$ to $\arc{v}{u}$, then set $\lb{a}=\ld{u}$.\\
    \tcc{This restores the original value of $\ld{u}$, but now $\ld{v}$ has decreased.}
    Call \code{check-dec}($v$).\\
    Return.
}
\For {up to $C/\alpha$ arcs $b=\arc{w}{u}\in \incut{u}$, in decreasing order of $\lb{b}$}
{
    \tcc{The label is much higher than the current load.}
    \uIf {$\lv{\lb{b}}\geq \lv{\ld{u}}+2$}
    {
        Set $\lb{b}=\ld{u}$.
    }
    \tcc{There are no arcs that satisfy the above conditions; all arc labels are close to the current load.}
    \Else
    {
        Return.
    }
}
\end{algorithm}

We have four subroutines; pseudocode is presented in Algorithms \ref{alg:insert3} through \ref{alg:checkdec3}. \code{insert($e=(u,v)$)} inserts an undirected edge $e$ into the data structure. \code{delete($e$)} deletes a directed edge $e$ from the data structure. \code{check-inc($v$)} is called whenever an additional edge is oriented into vertex $v$, either by an edge insertion or an arc flip. \code{check-dec($u$)} is called whenever an edge is oriented out of vertex $u$, either by an edge deletion or an arc flip.

For each edge update, the data structure uses these subroutines to maintain a $(\bigO{\alpha}, \bigO{1})$-locally optimal orientation.

Whenever the load on vertex $v$ increases, we call \code{check-inc($v$)} on that vertex.
\code{check-inc($v$)} ensures that after vertex $v$ increases in load, (a) all arcs $\arc{u}{v}$ are still approximately locally optimal, and (b) all arc labels $\lb{\arc{u}{v}}$ are still close to $\ld{v}$. To accomplish this, \code{check-inc($v$)} updates the $C/\alpha$ arcs with lowest label; these arcs are closest to violating the above conditions when $\ld{v}$ increases. If the local optimality condition is violated, then the arc is flipped to orient in the other direction. This returns $\ld{v}$ back to its original value, but increases $\ld{u}$ instead, so we recurse on \code{check-inc($u$)}. If the arc isn't flipped, then its label $\lb{\arc{u}{v}}$ is reset to $\ld{v}$.

A similar process occurs for \code{check-dec($u$)}. Whenever vertex $u$ decreases in load, \code{check-dec($u$)} ensures (a) all arcs $\arc{u}{v}$ are still approximately locally optimal, and (b) all arc labels $\lb{\arc{w}{u}}$ are still close to $\ld{u}$. To accomplish this, \code{check-dec($u$)} checks if the arc $\arc{u}{v}$ with highest label should be flipped. If it is, then we recurse on \code{check-dec($w$)} in a similar manner. If not, then $C/\alpha$ arcs with highest labels $\lb{\arc{w}{u}}$ have their labels reset back to $\ld{u}$.

\subsection{Correctness analysis}

In this section, we show that data structure maintains a $(\bigO{\alpha}, \bigO{1})$-locally optimal orientation. First, we show that the label of an arc is always close to the current load.

\begin{lemma}\label{lemma:worst-ineq-lowerbound-v}
    For all arcs $a=\arc{u}{v}$, $\lv{\ld{v}}\leq \lv{\lb{a}}+4$.
\end{lemma}

\begin{proof}
    Fix $a$.     
    We call $a$ \ita{bad} if it violates the given inequality, and \ita{dangerous} if
    \[
        \lv{\ld{v}}\geq \lv{\lb{a}}+2.
    \]
    We show that $a$ never becomes bad. Note that $a$ must be dangerous before it becomes bad.
    
    When $a$ is initially added to the orientation, or $a$ is relabeled, we have $\lb{a}=\ld{v}$, at which point it is neither bad nor dangerous.

    Suppose $a$ becomes dangerous. The loop in \code{check-inc($v$)} can now process $a$. Let $\lambda=\ld{v}$ at this moment when $a$ became dangerous. $\lv{\ld{v}}$ must increase by at least two more levels before $a$ becomes bad. Thus, $\ld{v}$ must increase by at least $(\alphamore \lambda + 1)-\lambda=\alpha\lambda+1=\bigOmega{\alpha\lambda}$. Each time $\ld{v}$ increases (i.e., at the end of a recursive chain of \code{check-inc}'s) is chance to process $a$, and consequently reset $\lab{a}=\ld{v}$.

    \code{check-inc($v$)} processes arcs $b\in\incut{v}$ with minimum label $\lb{b}$. An arc $b$ is processed before $a$ only if $\lb{b}\leq \lb{a}$. When $a$ becomes dangerous, there are at most $\indeg{v} = \lambda\wt{v}$ arcs $b$ with $\lb{b} \leq \lb{a}$. The number of arcs cannot increase while $a$ is dangerous, since new or updated labels are assigned value $\ld{v}>\lb{a}$. 

    Each time \code{check-inc($v$)} is called, if $a$ is not processed, then instead we process $C/\alpha$ arcs with lower label. Each of those arcs has its label set to greater than $\lb{a}$. Thus, the total number of arcs $b \in \incut{v}$ with $\lb{b} \leq \lb{a}$ decreases by $C / \alpha$.

    Suppose by contradiction that $a$ goes from dangerous to bad. In order to become bad, $a$ must not have been processed while it was dangerous. From the time $a$ became dangerous, the load must have increased by at least $\bigOmega{\alpha\lambda}$. Each decrease to $\ld{v}$ is by $1/\wt{v}$, and processes $C/\alpha$ arcs $b$ with $\lb{b}\leq\lb{a}$. Altogether, we must process at least
    \[
        \frac{\bigOmega{\alpha\lambda}\cdot C/\alpha}{1/\wt{v}}
        >
        \lambda\wt{v}
    \]
    arcs. which is the in-degree of $v$ when $a$ became dangerous. However, the number of arcs $b$ with $\lb{b}\leq \lb{a}$ was bounded above by this in-degree, a contradiction.

    Thus, $a$ never goes from dangerous to bad.
\end{proof}

\begin{lemma}\label{lemma:worst-ineq-upperbound-v}
    For all arcs $a=\arc{u}{v}$, $\lv{\lb{a}}\leq \lv{\ld{v}}+4$.
\end{lemma}

\begin{proof}
    The proof is similar to that of Lemma \ref{lemma:worst-ineq-lowerbound-v}.

    Fix $a$. We call $a$ \ita{bad} if it violates the given inequality, and \ita{dangerous} if
    \[
        \lv{\lb{a}}\geq \lv{\ld{v}}+2.
    \]

    We show that $a$ never becomes bad. Note that $a$ must be dangerous before it becomes bad.
    
    When $a$ is initially added to the orientation, or $a$ is relabeled, $\lb{a}=\ld{v}$, so it is neither bad nor dangerous. 

    Suppose $a$ becomes dangerous. The loop in \code{check-dec($v$)} can now process $a$. Let $\lambda=\ld{v}$ at the moment $a$ becomes dangerous. $\lv{\ld{v}}$ must decrease by at least two more levels before $a$ becomes bad. Thus, $\ld{v}$ must decrease by at least $\lambda - \frac{1}{\alphamore}(\lambda-1) = \frac{\alpha}{\alphamore}\lambda+\frac{1}{\alphamore}=\bigOmega{\alpha\lambda}$. Each time $\ld{v}$ decreases (i.e., at the end of a recursive chain of \code{check-dec}'s) is a chance to process $a$, and consequently reset $\lb{a}=\ld{v}$.

    \code{check-dec($v$)} processes arcs $b\in\incut{v}$ with maximum label $\lb{b}$. An arc $b$ is processed before $a$ only if $\lb{b}\geq \lb{a}$. When $a$ becomes dangerous, there are at most $\indeg{v}=\lambda\wt{v}$ arcs $b$ with $\lb{b}\geq \lb{a}$. The number of arcs cannot increase while $a$ stays dangerous, since new or updated labels are assigned value $\ld{v}<\lb{a}$.

    Each time \code{check-dec($v$)} is called, if $a$ is not processed, then instead we process $C/\alpha$ arcs with higher label. Each of those arcs has its label set to less than $\lb{a}$. Thus, the total number of arcs $b \in \incut{v}$ with $\lb{b} \geq \lb{a}$ decreases by $C / \alpha$.

    Suppose by contradiction that $a$ goes from dangerous to bad. 
    In order to become bad, $a$ must not have been processed while it was dangerous.
    From the time $a$ became dangerous, the load must have decreased by at least $\bigOmega{\alpha\lambda}$. 
    Each decrease to $\ld{v}$ is by $1/\wt{v}$, and processes $C/\alpha$ arcs $b$ with $\lb{b} \geq \lb{a}$. Altogether, we must process at least
    \begin{align*}
        \frac{\bigOmega{\alpha \lambda} \cdot C/\alpha}{1/\wt{v}}
        >
        \lambda\wt{v},
    \end{align*}
    arcs, which is the in-degree of $v$ when $a$ became dangerous. However, the number of arcs $b$ with $\lb{b}\geq\lb{a}$ was bounded above by this in-degree, a contradiction.

    Thus, $a$ never goes from dangerous to bad.
\end{proof}

Now that we have bounded $\lb{a}$ by $\ld{v}$, we also show a bound with $\ld{u}$.

\begin{lemma}\label{lemma:worst-ineq-upperbound-u}
    For all arcs $a=\arc{u}{v}$, $\lv{\lb{a}}\leq \lv{\ld{u}} + 3$.
\end{lemma}

\begin{proof}
    Fix $a=\arc{u}{v}$. When we initially add an edge to the orientation, $\lv{\lb{a}}=\lv{\ld{v}}\leq \lv{\ld{u}+1/\wt{v}}\leq \lv{\ld{u}}+1$, so the inequality is satisfied.

    The inequality can possibly be violated when $\lb{a}$ is reset, or when $\ld{u}$ decreases. 

    $\lb{a}$ can be set in four places.

    First, $\lb{a}$ can be set when $a$ is flipped in \code{check-inc}. Suppose we flip arc $\arc{v}{u}$ to $a = \arc{u}{v}$. Let $\ldA{u}$ be the load of $u$ before the flip (respectively $\ldA{v}$), and $\ldB{u}$ be the load of $u$ after the flip (respectively $\ldB{v}$). We flip from $\arc{v}{u}$ to $\arc{u}{v}$ if $\lv{\ldA{v}}+2\leq \lv{\ldA{u}}$. After flipping, $\lv{\ld{v}}$ increases by at most $1$ and $\lv{\ld{u}}$ decreases by at most $1$, so we have
    \[
        \lv{\ldB{v}-1/\wt{v}}+2%
        =%
        \lv{\ldA{v}} + 2%
        \leq %
        \lv{\ldA{u}}=\lv{\ldB{u}+1/\wt{u}}
    \]
    Then, applying Lemma \ref{lemma:lv-smooth}, we find
    $\lv{\lb{a}}=\lv{\ldB{v}}\leq \lv{\ldB{u}}$.
    \begin{align*}
        \lv{\lb{a}} 
        &= \lv{\ldB{v}}
        \leq
        \lv{\ldA{v}} + 1
        \leq
        \lv{\ldA{u}} - 1
        \leq
        \lv{\ldB{u}}.
    \end{align*}
    
    Second, $\lb{a}$ can be set in the loop of \code{check-inc}. If the arc is not flipped and we set $\lb{a}$, then we have $\lv{\ld{u}}+2>\lv{\ld{v}}=\lv{\lb{a}}$.

    Third, $\lb{a}$ can be set if $a$ is flipped in \code{check-dec}. Suppose we flip arc $\arc{v}{u}$ to $a = \arc{u}{v}$. Let $\ldA{u}$ be the load of $u$ before the flip (respectively $\ldA{v}$), and $\ldB{u}$ be the load of $u$ after the flip (respectively $\ldB{v}$). We flip the arc if $\lv{\ldA{v}}+3\leq \lv{\lb{a}}$. Combining with Lemma \ref{lemma:worst-ineq-upperbound-v}, when we flip the arc we have $\lv{\ldA{v}}+3\leq \lv{\ldA{u}}+4$, or $\lv{\ldA{v}}\leq \lv{\ldA{u}}+1$. After the flip, $\ld{v}$ increases by $1/\wt{v}$ and $\ld{u}$ decreases by $1/\wt{v}$, and we have
    \[
        \lv{\ldB{v}-1/\wt{v}}=\lv{\ldA{v}}  \leq \lv{\ldA{u}}+1=\lv{\ldB{u}+1/\wt{u}}+1\\
    \]
    Then, applying Lemma \ref{lemma:lv-smooth}, we find
    \[
    \lv{\lb{a}}=\lv{\ldB{v}}\leq \lv{\ldA{v}}+1 \leq \lv{\ldA{u}}+2\leq\lv{\ldB{u}} + 3.
    \]

    Finally, $\lb{a}$ can be set in the loop of \code{check-dec}. In this case, $\lb{a}$ decreases, so the inequality stays satisfied.

    The inequality can also be violated when $\ld{u}$ decreases. Whenever this happens, \code{check-dec} is called. We choose an arc $b\in \outcut{u}$ such that $\lv{\lb{b}}$ is maximized; that is, $\lv{\lb{b}}\geq \lv{\lb{a}}$.

    If $b$ is flipped, then $\ld{u}$ increases back to its original value before it decreased, and the inequality stays satisfied.

    If $b$ is not flipped, then $\lv{\ld{u}}+3>\lv{\lb{b}}\geq \lv{\lb{a}}$.
\end{proof}

Finally, combining the above lemmas gives us $(\bigO{\alpha}, \bigO{1})$-local optimality.

\begin{lemma}\label{lemma:worst-locally-optimal}
    The orientation is always $(\bigO{\alpha}, \bigO{1})$-locally optimal.
\end{lemma}

\begin{proof}
    Combining Lemmas \ref{lemma:worst-ineq-lowerbound-v} and \ref{lemma:worst-ineq-upperbound-u}, for all arcs $a=\arc{u}{v}$ we have
    \[
        \lv{\ld{v}}\leq \lv{\lb{a}}+4\leq \lv{\ld{u}}+7.\qedhere
    \]
\end{proof}

\subsection{Complexity analysis}

We have shown that the data structure described above maintains a $(\bigO{\alpha}, \bigO{1})$-locally approximate orientation. We now discuss the running time and space complexity of the data structure on update.

\subsubsection{Label-maintaining}

In order to function, the data structure must, for each vertex, maintain a list of incoming arcs in ascending order of $\lb{a}$, and a list of outgoing arcs in descending order of $\lv{\lb{a}}$. The data structure also globally maintains a list of vertices in sorted order of $\ld{v}$ in the same fashion, which is necessary to obtain an estimate for the densest subgraph.

To order vertices in $\incut{v}$ in increasing order of $\lb{a}$, we organize the arcs in nested doubly linked lists. We keep arc $a=\arc{u}{v}$ in a linked list with all other arcs with the same label $\lb{a}$. These linked lists are then kept in an outer linked list, in order of label. Note that whenever we set $\lb{a}$, we set it to $\ld{v}$; thus, we always keep a pointer to the first linked list of label greater than or equal to $\ld{v}$. Since $\ld{v}$ only ever changes in multiples of $1/\wt{v}$, updating this pointer is easy. Thus, we can query an incoming arc with minimum label in constant time, and also insert an arc $\arc{u}{v}$ in constant time.

To order vertices in $\outcut{u}$ in decreasing order of $\lv{\lb{a}}$, we organize the arcs in a sorted binary search tree. We keep arcs $a=\arc{u}{v}$ in a linked list with all other arcs with the same value of $\lv{\lb{a}}$. We then store these linked lists in a sorted binary search tree, indexed by the level. Since $\lb{a}$ is at least $1/W$ and at most $\poly(nW)$, the number of keys in the tree is $\bigO{\log_{\alphamore}(nW)}=\bigO{\log(nW)/\alpha}$. Thus, we can query an outgoing arc with maximum level in $\bigO{\log(\log(nW)/\alpha)}$ time, and update $\lb{a}$ in $\bigO{\log(\log(nW)/\alpha)}$ time.

Thus, updating a label takes $\bigO{\log(\log(nW)/\alpha)}=\bigO{\log\log (nW) + \log(1/\eps)}$ time, since we must update the arc at both endpoints $u$ and $v$.

\subsubsection{Dealing with edge duplication}

Our data structure can handle parallel edges. Stored naively, we would keep track of each duplicate edge independently; this, however, increases the space required significantly. Instead, we simulate the edge duplication with minimal space overhead.

Consider edge $e=(u, v)$. In an orientation, $e$ can be oriented as $a_1=\arc{u}{v}$ or $a_2=\arc{v}{u}$. When we duplicate $e$, it is possible that some of the copies will be oriented as $\arc{u}{v}$, and the rest oriented as $\arc{v}{u}$.

Instead of treating each copy of an edge independently, we just record the number of times an edge is oriented in each particular direction. Furthermore, for each orientation $a_i$, we maintain a singular set of labels $\labc{a_i}{v}$, $\labc{a_i}{u}$ that will be used for all corresponding copies, rather than each copy having its own set of labels. Each orientation also stores a pointer to the opposite orientation.

We argue that doing so does not cause errors in the algorithm. Note that we only reset labels when doing so preserves the inequalities in Lemmas \ref{lemma:worst-ineq-lowerbound-v} to \ref{lemma:worst-ineq-upperbound-u}. In particular, if it is valid to update the labels for one copy of $a_i$, then it is valid to update the labels for any other copy.

Finally, we must modify the lists and trees that store arcs. Note that all copies of an arc have the same value, so instead of storing arcs separately, we have a single node for each orientation that stores the number of duplicates oriented as such. Removing an arc from the list corresponds to decrementing the counter, and adding an arc corresponds to incrementing the counter.

\subsubsection{Amortized running time analysis}

Let $S=\bigO{\log\log (nW) + \log(1/\eps)}$.

\begin{lemma}\label{lemma:amortized-insert}
    \code{insert} takes $\bigO{S/\alpha}$ amortized time.
\end{lemma}

\begin{proof}
    First, note that $\ld{v}$ only increases by $1/\wt{v}$ at a time (one more arc is oriented towards $v$).
    
    Second, observe that the running time for processing an edge insertion is proportional to the number of arcs considered in the loop of \code{check-inc} over all recursive calls.
    
    We use a credit scheme for our amortized analysis, where \code{insert}s distribute credits among incoming arcs, paying for later \code{check-inc} operations. Whenever $\ld{v}$ increases by $1/\wt{v}$, we spread $\bigO{S/\alpha}$ credits uniformly over the arcs in $\incut{v}$. Note that this occurs exactly once per edge insertion (although potentially after a chain of \code{check-inc}'s), so each insertion operation corresponds to a single load increase somewhere in the graph. Each such arc thus receives a credit of $\bigOmega{S/\alpha \incut{v}}$. One unit of credit will pay for a constant amount of work, so this adds an amortized cost of $\bigO{S/\alpha}$ to inserting one copy of an edge.
    
    Since processing an arc (by setting $\lb{a}$) takes $\bigO{S}$ time, we show that any arc processed by \code{check-inc} has accumulated at least $\bigO{S}$ units of credit since $\lb{a}$ was last set.
    
    Consider an edge $a=\arc{u}{v}$ processed in the loop. The arc label satisfies $\lv{\ld{v}}>\lv{\lb{a}} + 2$, and thus $\ld{v}>\alphamore\lb{a} + 1/\wt{v}$. Since $\lb{a}$ is a prior value of $\ld{v}$, $\incut{v}$ must have gained at least $\bigOmega{\alpha\indeg{v}}$ edges since the label was last set. Each edge gained contributed $\bigOmega{S/\alpha\indeg{v}}$ credit, and so $a$ has gained at least $S$ units of credit to pay for processing it.
\end{proof}

\begin{lemma}\label{lemma:amortized-delete}
    \code{delete} takes $\bigO{\log(nW)S/\alpha}$ amortized time.
\end{lemma}

\begin{proof}
    The running time for processing an edge deletion is proportional to the number of arcs considered in the loop of \code{check-dec} over all recursive calls.

    \code{delete} operations distribute credit among incoming arcs, which pay for later \code{check-dec} operations. Whenever $\ld{u}$ decreases by $1/\wt{u}$, we spread $\bigO{S/\alpha}$ credits uniformly over the arcs in $\incut{u}$. Again, this only occurs once per edge deletion, at the end of the recursive chain of \code{check-dec}s. Each incoming arc receives a credit of $\bigOmega{S/\alpha\indeg{u}}$. One unit of credit pays for a constant amount of work, so this adds an amortized cost of $\bigO{S/\alpha}$ to deleting one copy of an edge.

    Consider an edge $b=\arc{w}{u}$ processed in the loop. This arc satisfies $\lv{\lb{b}}\geq \lv{\ld{u}}+2$, and thus $\lb{b}\geq\alphamore\ld{u}+1/\wt{u}$. 

    Since $\lb{b}$ was last set to a prior value of $\ld{u}$, $\ld{u}$ must have gone down by at least $\lb{b}\cdot \alpha/\alphamore>\lb{a}\alpha/2$. Each incoming edge lost contributed $\bigOmega{S/\alpha\lb{b}}$ credit, so $b$ has gained at least $S$ units of credit to pay for processing it.

    We now bound the effect of edge flips on the runtime using the recursive depth. Note that for each recursive call to \code{check-dec}, $\ld{u}$ increases by at least one level. Since the load on each node is upper-bounded by $\poly(nW)$, the recursive depth is at most $\bigO{\log_{\alphamore}(nW)}=\bigO{\log(nW)/\alpha}$. Thus, we spend at most $\bigO{\log(nW)S/\alpha}$ time processing edge flips.
\end{proof}

\subsubsection{Worst-case running time analysis}

\begin{lemma}\label{lemma:worst-insert}
    \code{insert} takes $\bigO{\log(nW)S/\alpha^2}$ worst-case time.
\end{lemma}

\begin{proof}
    The running time is proportional to the number of arcs processed by all \code{check-inc} calls. Each arc processed takes $\bigO{S}$ time, due to label updating.

    In each of the loops, up to $\bigO{1/\alpha}$ arcs are processed. The load of the vertex in each subsequent recursive call goes down by an $\alphamore$ factor, so the recursive depth goes up to $\bigO{\log_{\alphamore}(nW)}=\bigO{\log(nW)/\alpha}$. Thus, a maximum of $\bigO{\log(nW)/\alpha^2}$ arcs are processed per \code{insert} or \code{delete}, and the worst-case running time is $\bigO{\log(nW)S/\alpha^2}$.
\end{proof}

\begin{lemma}\label{lemma:worst-delete}
    \code{delete} takes $\bigO{\log(nW)S/\alpha}$ worst-case time.
\end{lemma}

\begin{proof}
    The running time is proportional to the number of arcs processed by all \code{check-dec} calls. Each arc processed takes $\bigO{S}$ time, due to label updating.

    In each \code{check-dec} call, either we flip arc $a$, or the recursion terminates and we fix some edge labels. The level of the vertex in each subsequent recursive call goes down by at least 1, so the recursive depth goes up to $\bigO{\log_{\alphamore}(nW)}=\bigO{\log(nW)/\alpha}$. At the end of the recursion, $\bigO{1/\alpha}$ arcs are processed. Thus, the worst-case running time is $\bigO{\log(nW)S/\alpha}$.
\end{proof}

\section{Improved bounds in the low-density regime}\label{ch:threshold}

For vertex-weighted graphs, recall that our data structure above can maintain a subgraph with at least $\epsless \lpopt - \bigO{\log(nW)/\eps}$ density in $\softO{\log^2(nW)/\eps^2}$ amortized time and $\softO{\log^3(nW)/\eps^4}$ worst-case time per edge update.

Furthermore, applying the VWDSG-DDSG reduction, for (unweighted) directed graphs, we can maintain a subgraph with at least $\epsless \OPT - \bigO{\log(n)/\eps}$ density in $\softO{\log^3(n)/\eps^4}$ amortized time and $\softO{\log^4(n)/\eps^5}$ worst-case time per edge update.

This approximation is good when $\OPT$ is at least $\bigOmega{\log(n)/\eps^2}$. However, when $\OPT$ is lower, the additive factor is too large compared to the true value of $\OPT$, and the estimate becomes inaccurate. To obtain a pure $\epsless$-approximation in the low-density regime, we artificially inflate the density by duplicating edges $C\log(n)/\eps^2$ times for sufficiently large $C$. This gives a pure $\epsless$-approximation, but increases the cost of edge updates to $\softO{\log^4(n)/\eps^6}$ amortized time and $\softO{\log^5(n)/\eps^7}$ worst-case time.

Note that in the analysis in Lemmas \ref{lemma:amortized-delete}, \ref{lemma:worst-insert}, and \ref{lemma:worst-delete}, the running time is bounded by the recursive depth of \code{check-inc} and \code{check-dec} calls. When the density is low, the recursive depth is also low. Thus, we can obtain improved bounds when $\OPT$ is low.

With these observations in mind, we modify the data structure to produce good estimates only in the low-density regime, and detect if $\OPT$ becomes too high. Combined with our original data structure in the high-density regime (without edge duplication), the two algorithms running in parallel will produce better running times overall.

\begin{thm}\label{thm:threshold}
    Consider the problem of approximating the densest subgraph in a vertex-weighted graph with weights $w:V\to \Rgeo$, with the graph dynamically updated by edge insertions and deletions. Let $\eps, T>0$ be given, with $\eps$ sufficiently small and $T>\bigOmega{\log(nW)/\eps^2}$. If $\lpopt\leq T$, we can maintain an orientation (explicity) with maximum in-degree at most $\epsmore\lpopt + \bigO{\log(nW)/\eps}$, and a subgraph (implicitly) with density at least $\epsless\lpopt - \bigO{\log(nW)/\eps}$, in $\softO{\log(T)\log(nW)/\eps^2}$ amortized time per update and $\softO{\log(T)\log^2(nW)/\eps^4}$ worst-case time per update, where $W=\max_v \wt{v}$.
\end{thm}

Duplicating edges $\bigO{1/\alpha}$ times and setting $T=\bigO{\log^2(n)/\eps^4}$ is sufficient for the DDSG reduction.

\begin{corollary}\label{corollary:threshold-pure}
    Let $\eps>0$, and let $T=\bigO{\log^2(n)/\eps^4}$. For an (unweighted) directed graph dynamically updated by edge insertions and deletions, while $\OPT< \epsless T$, we can maintain a subgraph with density at least $\epsless\OPT$ in $\softO{\log^3(n)/\eps^5}$ amortized time per update and $\softO{\log^4(n)/\eps^7}$ worst-case time per update.
\end{corollary}

This data structure can also detect whether or not we are in the low-density regime. Thus, we run this algorithm alongside the original data structure, and select an answer based on which of the two estimates are accurate. Overall, this gives us improved bounds in the general case.

\begin{corollary}\label{corollary:threshold-ddsg}
    Let $\eps>0$. For an (unweighted) directed graph dynamically updated by edge insertions and deletions, we can maintain a subgraph with density at least $\epsless\OPT$ in $\softO{\log^3(n)/\eps^5}$ amortized time per update and $\softO{\log^4(n)/\eps^7}$ worst-case time per update.
\end{corollary}

\subsection{The data structure}

\subsubsection{Data structure overview}

The modification to the data structure incorporates a \ita{threshold} $T$ to the load. For all vertices $v$, we redefine a thresholded load as
\begin{align*}
    \loadT{v}=\min(\load{v}, T).
\end{align*}
Informally, we treat all vertices with true load greater than the threshold $T$ as equal.

Replacing all instances of $\ld{v}$ with $\ldT{v}$ has two effects. First, because we limit the maximum load to $T$, the recursive depth of any \code{check-inc} or \code{check-dec} call is at most $\bigO{\log_{\alphamore}(T)}=\bigO{\log(T)/\alpha}$, rather than $\bigO{\log(nW)/\alpha}$. Second, because only loads under $T$ will be accurate, our local optimality conditions will only apply in the case where all $\ld{v}\leq T$. However, this data structure only needs to function in the low-density regime, so setting $T$ to a large enough value is sufficient.

To be precise, we set $T=C_T\log^2 n/\eps^4=C_T/\alpha^2$ for $C_T$ sufficiently large. We require this data structure to be valid for graphs of density $\log n/\eps^2$ or lesser, as well as an extra $\log n/\eps^2$ factor to account for edge duplication. We restate our local optimality conditions and lemmas in terms of $\loadT{v}$ in order to prove the correctness of this scheme.

\subsubsection{The algorithm}

Let $C, T$ be sufficiently large.

\begin{algorithm}[p]
\caption{\code{insert}($e=(u,v)$)}
\label{alg:insert-threshold}
    Let $\ldT{u}\geq \ldT{v}$. (Otherwise, swap $u$ and $v$.)\\
    Orient $e$ as $a=\arc{u}{v}$ and add it to the orientation, then set $\lb{a}=\ldT{v}$.\\
    Call \code{check-inc}($v$).
\end{algorithm}

\begin{algorithm}[p]
\caption{\code{delete}($e=(u,v)$)}
\label{alg:delete-threshold}
    Let $e$ be oriented as $a=\arc{u}{v}$; delete $a$ from the orientation.\\
    Call \code{check-dec}($v$).
\end{algorithm}

\begin{algorithm}[p]
\caption{\code{check-inc}($v$)}
\label{alg:checkinc-threshold}
\tcc{This subroutine is called whenever $\ldT{v}$ has increased.}
\For {up to $C/\alpha$ arcs $a=\arc{u}{v}\in \incut{v}$, in increasing order of $\lb{a}$}
{
    \tcc{The label is much lower than the current load.}
    \uIf {$\lv{\ldT{v}}\geq\lv{\lb{a}}+2$}
    {
        \tcc{The arc is too imbalanced.}
        \uIf {$\lv{\ldT{u}}+2\leq\lv{\ldT{v}}$} 
        {
            Flip $a$ to $\arc{v}{u}$, then set $\lb{a}=\ldT{u}$.\\
            \tcc{This restores the original value of $\ldT{v}$, but now $\ldT{u}$ may have increased.}
            Call \code{check-inc}($u$).
            Return.
        }
        \Else
        {
            Set $\lb{a}=\ldT{v}$.
        }
    }
    \tcc{There are no arcs that satisfy the above conditions; all arc labels are close to the current load.}
    \Else
    {
        Return.
    }
}
\end{algorithm}

\begin{algorithm}[p]
\caption{\code{check-dec}($u$)}
\label{alg:checkdec-threshold}
\tcc{This subroutine is called whenever $\ldT{u}$ has decreased.}
Choose an arc $a=\arc{u}{v}\in\outcut{u}$ such that $\lv{\lb{a}}$ is maximized.\\
\If {$\lv{\ldT{u}}+3\leq \lv{\lb{a}}$}
{
    \tcc{The arc is too imbalanced.}
    Flip $a$ to $\arc{v}{u}$, then set $\lb{a}=\ldT{u}$.\\
    \tcc{This restores the original value of $\ldT{u}$, but now $\ldT{v}$ may have decreased.}
    Call \code{check-dec}($v$).\\
    Return.
}
\For {up to $C/\alpha$ arcs $b=\arc{w}{u}\in \incut{u}$, in decreasing order of $\lb{b}$}
{
    \tcc{The label is much higher than the current load.}
    \uIf {$\lv{\lb{b}}\geq \lv{\ldT{u}}+2$}
    {
        Set $\lb{b}=\ldT{u}$.
    }
    \tcc{There are no arcs that satisfy the above conditions; all arc labels are close to the current load.}
    \Else
    {
        Return.
    }
}
\end{algorithm}

Pseudocode is presented in Algorithms \ref{alg:insert-threshold} through \ref{alg:checkdec-threshold}; they are produced from replacing all instances of $\ld{v}$ with $\ldT{v}$ in Algorithms \ref{alg:insert3} through \ref{alg:checkdec3}.

\subsection{Correctness analysis}

This data structure makes identical decisions to the orientation as Algorithms \ref{alg:insert3} through \ref{alg:checkdec3}, albeit with respect to $\loadT{V}$ instead of $\load{v}$. Through the same logic, the same conditions are maintained with respect to $\loadT{v}$.

\begin{lemma}
    The orientation is always $(\bigO{\alpha}, \bigO{1})$-locally optimal with respect to $\ell_T$.
\end{lemma}

Recall that this implies that $\ldT{v}\leq (1+\bigO{\alpha})\ldT{u}+\bigO{1}$.

We certify that this local optimality condition also implies global optimality while $\OPT$ is small.

\begin{lemma}\label{lemma:localopt-threshold}
    Let $\alpha>0$ with $\alpha<c/\log(nW)$ for sufficiently small constant $c$. Let $\mu_T=\max_v\ldT{v}$. Suppose every arc $a=\arc{u}{v}$ is $(\alpha,\beta)$-locally optimal; that is, $a$ satisfies $\ldT{v}\leq (1+\alpha)\ldT{u}+\beta$. Then
    \[\mu_T\leq e^{\bigO{\alpha\log(nW)}}\left(\lpopt+\bigO{\sqrt{\frac{\log(nW)}{\alpha}}}\beta\right).\]
    In particular, given $\eps\in (0,1)$, if $\beta=\bigO{1}$ and $\alpha=c\eps^2/\log(nW)$ for sufficiently small constant $c$, then $\mu_T\leq \epsmore\lpopt+\bigO{\log(nW)/\eps}$.
\end{lemma}

\begin{proof}
    The proof follows identically as the proof of Lemma \ref{lemma:localopt}, except with $\load{v}$ replaced with $\loadT{v}$.
\end{proof}

It is not necessarily true that $\mu_T\geq \lpopt$ due to thresholding, and thus $\mu_T$ may not be a good estimate when $\lpopt$ is large. However, if all loads are under the threshold, then all $\loadT{v}=\load{v}$, and the original bounds apply.

Note that when $\OPT< \epsless T - \bigO{\log(nW)/\eps}$, we have $\mu_T< T$. Since $\mu_T=\max_v\min(\ld{v}, T)=\min(\max_v \ld{v}, T)$, it must be that $\max_v \ld{v}< T$. That is, if our estimate is significantly below the threshold, then the loads must all be below the threshold as well. Thus, $\mu_T$ is always a good estimate when $\lpopt$ is small.

\begin{corollary}
    When $\lpopt< \epsless T-\bigO{\log(nW)/\eps}$, $\mu_T$ is a $\epsmore$-approximation for $\lpopt$.
\end{corollary}

Equivalent statements can be made showing that this modified data structure maintains local optimality with respect to $\loadT{v}$.

\begin{lemma}\label{lemma:threshold-locally-optimal}
    The orientation is always $(\bigO{\alpha}, \bigO{1})$-locally optimal with respect to $\ell_T$. That is, for all arcs $a=\arc{u}{v}$, $\loadT{v}\leq (1+\bigO{\alpha}) \loadT{u}+\bigO{1}$.
\end{lemma}

The proofs are omitted because they follow exactly as the proofs of Lemmas \ref{lemma:worst-ineq-lowerbound-v} through \ref{lemma:worst-locally-optimal}, with $\loadT{v}$ replacing $\load{v}$. At a high level, the data structure still sets labels under the same conditions, and thresholding does not change the dangerous-bad argument.

\subsection{Complexity analysis}

Recall that for each vertex $v$, we maintain a sorted binary tree for arcs in $\outcut{v}$. The maximum value of $\lb{a}$ is $T$. Thus, it takes $\log(\log_{\alphamore}(T))=\bigO{\log(\log(T)/\alpha)}$ per label update.

Let $S=\bigO{\log(\log(T)) + \log(\log(nW)) + \log(1/\eps)}$.

\begin{lemma}\label{lemma:threshold-amortized-insert}
    \code{insert} takes $\bigO{S/\alpha}$ amortized time.
\end{lemma}

\begin{proof}
    The proof is identical to that of Lemma \ref{lemma:amortized-insert}.
\end{proof}

\begin{lemma}\label{lemma:threshold-amortized-delete}
    \code{delete} takes $\bigO{\log(T)S/\alpha}$ amortized time.
\end{lemma}

\begin{proof}
    The proof is almost identical to that of Lemma \ref{lemma:amortized-delete}. except the recursive depth of \code{check-dec} operations is bounded by $\bigO{\log_{\alphamore}(T)}=\bigO{\log(T)/\alpha}$. 
\end{proof}

\begin{lemma}\label{lemma:threshold-worst}
    \code{insert} takes $\bigO{\log(T)S/\alpha^2}$ worst-case time and \code{delete} takes $\bigO{\log(T)S/\alpha}$ worst-case time.
\end{lemma}

\begin{proof}
    The proof is almost identical to that of Lemma \ref{lemma:worst-insert} and \ref{lemma:worst-delete}. except the recursive depth of \code{check-inc} and \code{check-dec} operations is bounded by $\bigO{\log_{\alphamore}(T)}=\bigO{\log(T)/\alpha}$. 
\end{proof}

\section*{Acknowledgements}

We thank Kent Quanrud for his guidance and detailed feedback on this manuscript.

\printbibliography

\clearpage

\end{document}